\documentclass[11pt,a4paper]{article}
\usepackage[T1]{fontenc}
\usepackage[utf8]{inputenc}
\usepackage{lmodern}
\usepackage{microtype}
\usepackage[
  a4paper,
  top=27mm,
  bottom=27mm,
  left=27mm,
  right=27mm
]{geometry}
\usepackage{amsmath,amssymb,amsthm,mathtools}
\usepackage{mathrsfs}
\usepackage{bm}
\usepackage{booktabs}
\usepackage{tabularx}
\usepackage{enumitem}
\usepackage{xcolor}
\usepackage{authblk}
\usepackage[
  colorlinks=true,
  linkcolor=blue!55!black,
  citecolor=blue!55!black,
  urlcolor=blue!55!black
]{hyperref}

\hypersetup{
  pdftitle={
    Deterministic Minimum-Output-Entropy Nonadditivity
    via Haagerup's Inequality and Near-Free Permutation Representations
  },
  pdfauthor={Guocheng Zhen, Chengkai Zhu, Ranyiliu Chen, and Xin Wang},
  pdfsubject={
    Exact quadratic certificate and asymptotic fixed-parameter algorithmic
    construction of minimum-output-entropy nonadditivity
  },
  pdfkeywords={
    minimum output entropy,
    additivity violation,
    free probability,
    Haagerup inequality,
    permutation lifts,
    branch Gram matrices,
    Holevo superadditivity,
    fixed-parameter algorithms
  }
}

\usepackage[nameinlink,noabbrev]{cleveref}
\setlist{nosep}
\allowdisplaybreaks
\theoremstyle{plain}
\newtheorem{theorem}{Theorem}[section]
\newtheorem{proposition}[theorem]{Proposition}
\newtheorem{lemma}[theorem]{Lemma}
\newtheorem{corollary}[theorem]{Corollary}

\theoremstyle{definition}

\theoremstyle{remark}
\newtheorem{remark}[theorem]{Remark}

\newcommand{\C}{\mathbb C}
\newcommand{\R}{\mathbb R}

\newcommand{\F}{\mathbb F}
\newcommand{\one}{\mathbf 1}

\newcommand{\Tr}{\operatorname{Tr}}
\newcommand{\tr}{\operatorname{tr}}
\newcommand{\spec}{\operatorname{spec}}
\newcommand{\supp}{\operatorname{supp}}
\newcommand{\rank}{\operatorname{rank}}

\newcommand{\Dens}{\mathcal D}
\newcommand{\B}{\mathcal B}

\newcommand{\Hmin}{H_{\min}}
\newcommand{\HS}{\mathrm{HS}}

\newcommand{\ket}[1]{\lvert #1\rangle}
\newcommand{\bra}[1]{\langle #1\rvert}
\newcommand{\proj}[1]{\lvert #1\rangle\!\langle #1\rvert}

\newcommand{\ad}{\operatorname{Ad}}

\newcolumntype{Y}{>{\raggedright\arraybackslash}X}

\title{
  \textbf{Deterministic Minimum-Output-Entropy Nonadditivity via Haagerup's Inequality and Near-Free Permutation Representations}
}

\author[1]{Guocheng Zhen\thanks{zhenguocheng0814@gmail.com}}
\author[2]{Chengkai Zhu\thanks{zhuchengkai7@gmail.com}}
\author[3]{Ranyiliu Chen\thanks{chenranyiliu@quantumsc.cn.}}
\author[1]{Xin Wang\thanks{felixxinwang@hkust-gz.edu.cn.}}
\affil[1]{\small Thrust of Artificial Intelligence, Information Hub,\par The Hong Kong University of Science and Technology (Guangzhou), Guangdong 511453, China}
\affil[2]{\small QudeLeap Research, Shanghai 200030, China}
\affil[3]{\small Quantum Science Center of Guangdong-Hong Kong-Macao Greater Bay Area, Shenzhen 518045, China}
\date{}

\begin{document}

\maketitle

\begin{abstract}
We give a deterministic realization of the finite-dimensional quadratic
certificate underlying Collins's mixed-unitary proof of
minimum-output-entropy nonadditivity. For every fixed integer $K\ge2$ and rational $\eta>0$ satisfying
$\log K>2(3+\eta)^2$, a deterministic polynomial-time algorithm, for
every sufficiently large target size $N$, outputs $K$ permutations on
$N'=N+o_{K,\eta}(N)$ points.  Restricting their permutation matrices to
the nontrivial standard representation yields real orthogonal
Stinespring blocks and a channel
$\Phi_{N'}:M_{N'-1}(\C)\to M_K(\C)$ such that
\[
   2\Hmin(\Phi_{N'})
   -\Hmin(\Phi_{N'}^{\otimes2})
   \ge
   \frac{\log K}{K}
   -2\log\left(1+\frac{(3+\eta)^2}{K}\right)
   >0.
\]
The construction combines Haagerup's length-two inequality with the
simultaneous deterministic spectral approximation of O'Donnell and Wu.

We further show that the constant $3$ is asymptotically sharp on the
relevant Hermitian zero-diagonal coefficient class and that the finite
spectral transfer is nearly saturated, thereby isolating the finer
geometry of the full output body as the natural next level of refinement
beyond the scalar-radius method.  Finally, a standard covariant extension
converts the same deterministic entropy gap exactly into self-tensor
superadditivity of the one-shot Holevo quantity.
\end{abstract}

\newpage
\tableofcontents

\section{Introduction}\label{sec:introduction}

For finite-dimensional quantum channels $\Phi$ and $\Psi$, product inputs
always give
\[
   \Hmin(\Phi\otimes\Psi)
   \le
   \Hmin(\Phi)+\Hmin(\Psi),
\]
where $\Hmin$ denotes the minimum output von Neumann entropy. The minimum-output-entropy
additivity conjecture asked whether equality must hold for every pair of
channels.  This order-one problem is particularly important because of its
connection with classical communication over quantum channels. The Holevo--Schumacher--Westmoreland
theorem expresses the unassisted classical capacity through the
regularization of the Holevo quantity
\cite{Holevo1998,SchumacherWestmoreland1997}.  Shor subsequently proved
that, as universal statements over finite-dimensional quantum channels,
minimum-output-entropy additivity at $p=1$ is equivalent to additivity of
the Holevo quantity, as well as to additivity and strong superadditivity of
entanglement of formation \cite{Shor2004}.  Universal additivity would
therefore have reduced the corresponding regularized capacity formula to a
single-letter one.

Hastings disproved the $p=1$ conjecture by a random finite-dimensional
construction \cite{Hastings2009}.  In the mixed-unitary formulation,
one-copy outputs remain close to maximally mixed, whereas a maximally
entangled input to a channel and its complex conjugate creates a
lower-entropy two-copy output through collisions of matching unitary
branches.  Subsequent work clarified and sharpened this probabilistic and
geometric mechanism
\cite{FukudaKingMoser2010,BrandaoHorodecki2010,AubrunSzarekWerner2011,
BelinschiCollinsNechita2016}.  Collins later gave a particularly economical
free-probabilistic proof: strong asymptotic freeness replaces large Haar
unitaries by free Haar unitaries, while Haagerup's inequality controls the
resulting quadratic free-group expressions
\cite{CollinsMale2014,Collins2018}.  A key feature of Collins's argument is
that randomness is used only to produce a finite unitary tuple satisfying a
uniform family of quadratic operator-norm estimates; once such a tuple is
available, the entropy comparison is deterministic.

Away from $p=1$, the development has followed a different pattern.  Werner
and Holevo gave an early explicit violation for sufficiently large R\'enyi
orders, and Grudka, Horodecki, and Pankowski later obtained constructive
counterexamples for every $p>2$
\cite{WernerHolevo2002,GrudkaHorodeckiPankowski2010}.  Hayden and Winter,
by contrast, proved through random constructions that additivity fails for
every $p>1$ \cite{HaydenWinter2008}, whereas Derksen and Lovitz recently
obtained constructive counterexamples throughout the entire range $p>1$
\cite{DerksenLovitz2026}.  At the opposite end of the R\'enyi scale, Cubitt
et al.\ established nonadditivity at $p=0$, and hence for sufficiently small
positive $p$ \cite{CubittEtAl2008}; more recently, Leung, Lovitz, and Wu
proved randomized nonadditivity for $3/4<p<1$ and for $0\le p<1/4$, leaving
$[1/4,3/4]$ as the remaining unresolved interval inside $0<p<1$
\cite{LeungLovitzWu2026}.  At the von Neumann point, Nema and Sen replaced
Haar randomness by approximate unitary designs, obtaining a partial
derandomization of Hastings's construction rather than a deterministic
channel construction \cite{NemaSen2022}.

The question pursued here is whether Collins's mixed-unitary argument
itself admits an asymptotic fixed-parameter deterministic realization.
We answer it affirmatively by combining the scalar degree-two reduction
behind Collins's Haagerup estimate with the deterministic permutation-lift
machinery of O'Donnell and Wu
\cite{BordenaveCollins2019,ODonnellWu2020}.  Their simultaneous spectral
approximation yields a finite permutation tuple satisfying the uniform
quadratic estimates required in Collins's one-copy argument.  Restricting
these permutations to the nontrivial standard representation produces real
orthogonal Stinespring blocks, so the usual product--conjugate Bell witness
becomes a genuine tensor-square violation.  Thus, for fixed admissible
parameters $K$ and $\eta$, one deterministic polynomial-time algorithm
produces $p=1$ counterexamples for every sufficiently large target size.

Beyond this deterministic transfer, we identify the underlying quadratic
radius exactly, show that the Haagerup constant $3$ is asymptotically sharp
for the relevant coefficient class and that the finite spectral transfer is
nearly saturated, and convert the resulting entropy gap into one-shot
Holevo superadditivity by a standard covariant extension.  These results
also expose an intrinsic limitation of the scalar-radius method and motivate
the finer output-body optimization studied later.  The construction is
algorithmic and asymptotic rather than closed-form: we do not claim an
explicit finite threshold or a practically small instance.  

\section{Main results}\label{sec:main-results}

We first state the precise quantifiers of the construction.  The output
dimension $K$ and rational tolerance $\eta$ are fixed once and for all;
the target lift size $N$ is the algorithmic input and tends to infinity.
Throughout, all logarithms are natural.  We say that a channel
\emph{admits a real Stinespring realization} if it has a Stinespring
isometry whose matrix entries are real in specified orthonormal bases.

\begin{theorem}[Asymptotic fixed-parameter algorithmic construction]\label{thm:main}
Let $K\ge2$ be an integer and let $\eta\in\mathbb Q_{>0}$ satisfy
\begin{equation}\label{eq:main-threshold}
   \log K>2(3+\eta)^2.
\end{equation}
Then there exist an integer $N_0=N_0(K,\eta)$ and one deterministic
algorithm $\mathcal A_{K,\eta}$ with the following property.  For every
integer $N\ge N_0$, on input $N$ the algorithm runs in time polynomial in
$N$, with the polynomial allowed to depend on the fixed parameters
$K$ and $\eta$, and outputs an integer
$N'=N'_{K,\eta}(N)$ together with permutations
\[
   \sigma_1,\ldots,\sigma_K\in S_{N'}.
\]
The output sizes satisfy
\begin{equation}\label{eq:main-output-size}
   N\le N'_{K,\eta}(N),
   \qquad
   N'_{K,\eta}(N)-N=o_{K,\eta}(N)
   \quad(N\to\infty).
\end{equation}
The output permutations determine, canonically up to input orthogonal
equivalence, a quantum channel admitting a real Stinespring realization,
\[
   \Phi_{N'}:M_{N'-1}(\C)\longrightarrow M_K(\C),
\]
for which
\begin{align}
   2\Hmin(\Phi_{N'})-\Hmin(\Phi_{N'}^{\otimes2})
   &\ge
   \frac{\log K}{K}
   -2\log\left(1+\frac{(3+\eta)^2}{K}\right)
   \label{eq:main-gap-strong}\\
   &\ge
   \frac{\log K-2(3+\eta)^2}{K}
   >0.
   \label{eq:main-gap}
\end{align}
\end{theorem}

The quantifier order is therefore
\[
   \exists\,\mathcal A_{K,\eta}\;
   \exists\,N_0(K,\eta)\;
   \forall\,N\ge N_0(K,\eta).
\]
This is an asymptotic fixed-parameter algorithmic statement: after
$K$ and $\eta$ are fixed, one algorithm handles every sufficiently large
input size $N$.  No closed-form permutation tuple, explicit numerical
value of $N_0$, or practically small instance is asserted.  Restricting
$\eta$ to $\mathbb Q_{>0}$ gives the tolerance a finite description and
causes no loss: whenever a positive real tolerance satisfies
\eqref{eq:main-threshold}, so does some smaller positive rational one.

\begin{remark}[Channel construction and input orthogonal equivalence]
\label{rem:channel-construction}
For each output permutation, let $P_{\sigma_i}$ be its permutation matrix
on $\C^{N'}$.  Since $P_{\sigma_i}\one=\one$, the complex codimension-one
subspace
\[
   \mathcal H_{N'}^0
   :=
   \left\{x\in\C^{N'}:\sum_{a=1}^{N'}x_a=0\right\}
   =\one^\perp
\]
is invariant under every $P_{\sigma_i}$.  It is the
complexification of the real standard representation
\[
   \left\{x\in\R^{N'}:\sum_{a=1}^{N'}x_a=0\right\}.
\]
The restricted operators
\[
   \widetilde U_i
   :=P_{\sigma_i}\big|_{\mathcal H_{N'}^0}
\]
preserve this real form and are represented by real orthogonal matrices
in any real orthonormal basis.  They canonically define the Stinespring
isometry
\[
   \widetilde Vx
   :=
   \frac1{\sqrt K}\sum_{i=1}^K
      \widetilde U_i x\otimes\ket{i},
   \qquad
   x\in\mathcal H_{N'}^0,
\]
as well as its complementary channel on
$\B(\mathcal H_{N'}^0)$.

To express this channel on $M_{N'-1}(\C)$, choose any real orthonormal
identification
\[
   Q:\C^{N'-1}\longrightarrow\mathcal H_{N'}^0
\]
and set
\[
   U_i^Q:=Q^*P_{\sigma_i}Q\in O(N'-1).
\]
The corresponding Stinespring isometry and complementary channel are
\begin{align}
   V^Qx
   &:=
   \frac1{\sqrt K}\sum_{i=1}^K U_i^Qx\otimes\ket{i},
   \label{eq:channel-stinespring}\\
   \Phi_{N'}^Q(X)
   &:=
   \frac1K
   \left[\Tr\bigl(U_i^QX(U_j^Q)^*\bigr)\right]_{i,j=1}^K.
   \label{eq:channel-main-results}
\end{align}
Indeed, $(V^Q)^*V^Q=I_{N'-1}$.

If $Q'$ is another real orthonormal identification, then $Q'=QO$ for
some $O\in O(N'-1)$.  Hence
$U_i^{Q'}=O^*U_i^Q O$ and
\begin{equation}\label{eq:input-orthogonal-equivalence}
   \Phi_{N'}^{Q'}
   =
   \Phi_{N'}^Q\circ\ad_O,
   \qquad
   \ad_O(X):=OXO^*.
\end{equation}
Thus the channel is determined canonically up to precomposition by an
orthogonal conjugation on the input.  This equivalence leaves both its
one-copy minimum output entropy and the minimum output entropy of its
tensor square unchanged.

Because the matrices $U_i^Q$ are real, $V^Q$ has real matrix entries in
the chosen bases.  Entrywise channel conjugation therefore fixes
$\Phi_{N'}^Q$, so
\[
   \overline{\Phi_{N'}^Q}=\Phi_{N'}^Q.
\]
Consequently, the product--conjugate Bell witness used below is exactly a
witness for $(\Phi_{N'}^Q)^{\otimes2}$.
\end{remark}

The following corollary fixes one admissible pair of parameters and gives
an exact rational certificate for the simpler entropy-gap lower bound in
\eqref{eq:main-gap}.  It does not specify the threshold $N_0$ or the
permutations produced beyond that threshold.

\begin{corollary}[A concrete choice of fixed parameters]\label{cor:concrete}
Take
\[
   K=2^{26},
   \qquad
   \eta=10^{-3}.
\]
For every sufficiently large input $N$, the algorithm
$\mathcal A_{K,\eta}$ produces a channel admitting a real Stinespring
realization and satisfying
\[
   \Hmin(\Phi^{\otimes2})<2\Hmin(\Phi).
\]
More precisely,
\[
   2\Hmin(\Phi)-\Hmin(\Phi^{\otimes2})
   >
   \frac{741757}{121500000\cdot2^{26}}
   >9.09\times10^{-11}\ \text{nats}.
\]
\end{corollary}

The quadratic radius appearing in the proof is not merely a sufficient
upper bound.  If $G_x$ is the Gram matrix of the branch orbit
$U_1x,\ldots,U_Kx$, then Theorem~\ref{thm:gram-purity} below proves
\[
   \Gamma_K(U_1,\ldots,U_K)
   =\max_{\|x\|=1}\|G_x-I_K\|_{\HS}
   =K\max_{X\in\Dens(\C^{N'-1})}
      \left\|\Phi(X)-\frac{I_K}{K}\right\|_{\HS},
\]
and identifies the maximum output purity exactly.  The same analysis shows
that the finite tuple produced by the deterministic spectral transfer is
nearly saturated between
$(3K-4)/\sqrt{K(K-1)}-\eta$ and $3+\eta$, and that the full value of
$\Gamma_K$ is already attained along a direction generated by a pure
channel output.  Section~\ref{sec:discussion} proves these assertions and
formulates the remaining finer optimization in terms of the full free
output-body support function.

The operational consequence can also be stated at the level of classical
communication.  The next corollary is obtained by applying the standard
flagged covariant-extension method to the real self-tensor channel; a
self-contained proof is given in Section~\ref{sec:holevo-conversion}.

\begin{corollary}[Associated Holevo-superadditive channel]
\label{cor:holevo-main}
Assume the hypotheses of Theorem~\ref{thm:main} and suppose that
$K=2^q$ for some integer $q\ge1$.  For every channel
$\Phi_{N'}$ produced by the theorem, there is an associated channel
\[
   \widehat\Phi_{N'}:
   M_{K^2(N'-1)}(\C)\longrightarrow M_K(\C)
\]
admitting a real Stinespring realization such that
\begin{align}
   \chi(\widehat\Phi_{N'}^{\otimes2})
   -2\chi(\widehat\Phi_{N'})
   &=
   2\Hmin(\Phi_{N'})-\Hmin(\Phi_{N'}^{\otimes2})
   \label{eq:main-holevo-equality}\\
   &\ge
   \frac{\log K}{K}
   -2\log\left(1+\frac{(3+\eta)^2}{K}\right)
   >0.
   \label{eq:main-holevo-gap}
\end{align}
Consequently, the unassisted classical capacity of
$\widehat\Phi_{N'}$ is strictly larger than its one-shot Holevo quantity.
\end{corollary}

\section{Preliminaries}\label{sec:preliminaries}

This section mainly records the operator-algebraic and free-probabilistic
facts used in the proof.  For
broader introductions to free probability, see
\cite{VoiculescuDykemaNica1992,NicaSpeicher2006,MingoSpeicher2017}; for
operator-algebraic background, see \cite{BrownOzawa2008}.  All algebras
and Hilbert spaces are complex unless stated otherwise.  The symbol $\Tr$
denotes the unnormalized matrix trace and
$\tr_n:=n^{-1}\Tr$ the normalized trace on $M_n(\C)$. For $T\in\B(M_n(\mathbb{C}))$, we denote norms
\[
   \|T\|:=\sup_{\|x\|=1}\|Tx\|,
   \qquad
   \|T\|_{\HS}:=\bigl(\Tr(T^*T)\bigr)^{1/2}.
\]
If $T=T^*$, then
\begin{equation}\label{eq:self-adjoint-norm-spectrum}
   \|T\|=\max\{|t|:t\in\spec(T)\}.
\end{equation}
The normalized tracial $L^2$ norm on $M_n(\C)$ is
$\|T\|_{2,\tr_n}=n^{-1/2}\|T\|_{\HS}$; the channel estimates below use
the unnormalized Hilbert--Schmidt norm.

\subsection{Operator and tracial conventions}

For a Hilbert space $\mathcal H$, let $\B(\mathcal H)$ denote the bounded
operators on $\mathcal H$.  A concrete unital $C^*$-algebra is a
norm-closed $*$-subalgebra of $\B(\mathcal H)$ containing the identity;
by the Gelfand--Naimark theorem, every abstract $C^*$-algebra admits such
a faithful representation.

A \emph{tracial $C^*$-probability space} is a unital $C^*$-algebra
$\mathcal A$ equipped with a faithful tracial state
$\tau:\mathcal A\to\C$.  Thus $\tau$ is linear,
$\tau(a^*a)\ge0$, $\tau(1)=1$,
$\tau(a^*a)=0$ only when $a=0$, and
\[
   \tau(ab)=\tau(ba),
   \qquad a,b\in\mathcal A.
\]
The trace plays the role of expectation; traciality gives cyclicity of
mixed moments without requiring the underlying variables to commute.  We
write
\[
   \|x\|_{2,\tau}:=\tau(x^*x)^{1/2}.
\]
The basic finite-dimensional example is $(M_n(\C),\tr_n)$.

For later reference, a family of unital $*$-subalgebras
$(\mathcal A_i)_{i\in I}$ is \emph{free} if
$\tau(a_1\cdots a_m)=0$ whenever
$a_j\in\mathcal A_{i_j}$, $\tau(a_j)=0$, and adjacent indices are
different.  A unitary $u$ is a \emph{Haar unitary} if
$\tau(u^m)=0$ for every $m\in\mathbb Z\setminus\{0\}$.

\subsection{Quantum channels and entropy conventions}
\label{subsec:qit-conventions}

For a finite-dimensional Hilbert space $\mathcal H$, write
\[
   \Dens(\mathcal H)
   :=
   \left\{
      \rho\in\B(\mathcal H):
      \rho\ge0,\ \Tr\rho=1
   \right\}
\]
for its state space.  A quantum channel
\[
   \Phi:
   \B(\mathcal H_{\mathrm{in}})
   \longrightarrow
   \B(\mathcal H_{\mathrm{out}})
\]
is a completely positive trace-preserving linear map.

For a state $\sigma$ and $0<p<\infty$, $p\ne1$, its R\'enyi entropy of
order $p$ is
\[
   H_p(\sigma)
   :=
   \frac{1}{1-p}\log\Tr(\sigma^p).
\]
We use the standard endpoint and continuous conventions
\[
   H_0(\sigma):=\log\rank(\sigma),
   \qquad
   H_1(\sigma)
   =H(\sigma)
   :=-\Tr(\sigma\log\sigma),
\]
where $0\log0:=0$.

For a finite-dimensional quantum channel $\Phi$, its minimum output
R\'enyi entropy of order $p$ is
\[
   H_{\min,p}(\Phi)
   :=
   \min_{\rho\in\Dens(\mathcal H_{\mathrm{in}})}
   H_p\bigl(\Phi(\rho)\bigr).
\]
At the von Neumann point we abbreviate
\[
   \Hmin(\Phi):=H_{\min,1}(\Phi).
\]

\subsection{The free-group operator model}

Let $\F_K=\langle g_1,\ldots,g_K\rangle$ be the free group on $K$
generators.  Every nonidentity element has a unique reduced expression
\[
   g=g_{i_1}^{\varepsilon_1}\cdots g_{i_d}^{\varepsilon_d},
   \qquad \varepsilon_j\in\{1,-1\},
\]
with no adjacent cancellation; its reduced word length is denoted by
$|g|=d$.  In particular, $g_i^{-1}g_j$ has length two when $i\ne j$.

Let $(\delta_h)_{h\in\F_K}$ be the canonical basis of
$\ell^2(\F_K)$.  The left regular representation is
\[
   \lambda:\F_K\longrightarrow
   \mathcal U\bigl(\ell^2(\F_K)\bigr),
   \qquad
   \lambda(g)\delta_h=\delta_{gh}.
\]
The reduced group $C^*$-algebra and the group von Neumann algebra are
\[
   C_r^*(\F_K):=C^*(\lambda(\F_K)),
   \qquad
   L(\F_K):=\lambda(\F_K)'',
   \qquad
   C_r^*(\F_K)\subseteq L(\F_K).
\]
The canonical trace on $L(\F_K)$ is
\[
   \tau(T):=\langle T\delta_e,\delta_e\rangle,
\]
and satisfies
$\tau(\lambda(g))=\mathbf1_{\{g=e\}}$.  It is a faithful normal tracial state.

\begin{lemma}[Canonical free Haar generators]
\label{lem:free-group-haar}
The operators
\[
   u_i:=\lambda(g_i),
   \qquad 1\le i\le K,
\]
form a free family of Haar unitaries in
$(L(\F_K),\tau)$.
\end{lemma}

This follows immediately from reduced-word uniqueness and the formula for
the canonical trace; see, for example,
\cite{VoiculescuDykemaNica1992,NicaSpeicher2006}.
For a finitely supported function $f:\F_K\to\C$, write
\[
   \lambda(f):=\sum_{g\in\F_K}f(g)\lambda(g).
\]

\subsection{Haagerup's inequality}

For $d\ge0$, let
\[
   S_d:=\{g\in\F_K:|g|=d\}.
\]

\begin{lemma}[Haagerup's inequality
{\cite[Lemma~1.4]{Haagerup1979}}]
\label{lem:haagerup}
If $f:\F_K\to\C$ is finitely supported and
$\supp(f)\subseteq S_d$, then
\begin{equation}\label{eq:haagerup-general}
   \|\lambda(f)\|
   \le(d+1)\|f\|_{\ell^2(\F_K)}
   =(d+1)\|\lambda(f)\|_{2,\tau}.
\end{equation}
\end{lemma}

Only $d=2$ is used below:
\begin{equation}\label{eq:haagerup-length-two}
   \|\lambda(f)\|\le3\|f\|_{\ell^2(\F_K)}
   \qquad\bigl(\supp(f)\subseteq S_2\bigr).
\end{equation}
The constant is independent of the number $K$ of free generators.

\subsection{Unitary polynomials and spectral control}

Let $Z_1,\ldots,Z_K$ be formal unitary variables, with
$Z_i^*Z_i=Z_iZ_i^*=1$.  Repeated cancellation of adjacent inverse pairs
produces a unique reduced unitary word.  Under
$Z_i\leftrightarrow g_i$ and $Z_i^*\leftrightarrow g_i^{-1}$, this is
exactly free-group reduction.  If a polynomial is written after reduction
as
\[
   p=\sum_w c_w w
\]
over distinct reduced words, its \emph{reduced degree} is the maximum word
length occurring with nonzero coefficient.  It is self-adjoint exactly
when $c_{w^*}=\overline{c_w}$ for every reduced word $w$.

For context, a sequence of tuples $a^{(n)}$ in tracial
$C^*$-probability spaces converges \emph{strongly in $*$-distribution} to
$a$ if, for every noncommutative $*$-polynomial $P$, both the tracial
moments and the operator norms converge:
\[
   \tau_n\bigl(P(a^{(n)})\bigr)\to\tau\bigl(P(a)\bigr),
   \qquad
   \|P(a^{(n)})\|\to\|P(a)\|.
\]
The strong asymptotic freeness theorem of Collins and Male says that
independent Haar unitary matrices converge almost surely in this sense to
the canonical free Haar family \cite[Theorem~1.5]{CollinsMale2014}; this
is the random approximation input in Collins's argument
\cite{Collins2018}.  The deterministic proof below instead uses the
finite-size spectral theorem of O'Donnell and Wu.

For nonempty compact sets $S,T\subset\mathbb R$, their Hausdorff distance
is
\[
   d_{\mathrm H}(S,T)
   :=
   \max\left\{
      \sup_{s\in S}\inf_{t\in T}|s-t|,
      \sup_{t\in T}\inf_{s\in S}|t-s|
   \right\}.
\]
If $X$ and $Y$ are self-adjoint elements of $C^*$-algebras and
$d_{\mathrm H}(\spec X,\spec Y)\le\eta$, then
\begin{equation}\label{eq:hausdorff-to-norm}
   \bigl|\|X\|-\|Y\|\bigr|\le\eta,
\end{equation}
because $t\mapsto|t|$ is $1$-Lipschitz and
\eqref{eq:self-adjoint-norm-spectrum} applies.

\subsection{Permutation representations and the O'Donnell--Wu theorem}

For $\sigma\in S_N$, let $P_\sigma$ be its permutation matrix and set
\[
   \ket{+}_N:=N^{-1/2}(1,\ldots,1)^T,
   \qquad
   \mathcal H_N^0:=\ket{+}_N^\perp.
\]
Every $P_\sigma$ is real orthogonal, fixes $\ket{+}_N$, and leaves
$\mathcal H_N^0$ invariant.  We call
$P_\sigma|_{\mathcal H_N^0}$ the nontrivial standard representation.
Relative to
$\C^N=\C\ket{+}_N\oplus\mathcal H_N^0$, one has
\[
   P_\sigma=1\oplus P_\sigma|_{\mathcal H_N^0}.
\]
Thus, for every scalar noncommutative $*$-polynomial $p$, evaluation on a
permutation tuple splits into the trivial scalar block $p(1,\ldots,1)$
and its restriction to $\mathcal H_N^0$.

The following is the scalar, permutation-only specialization of
Theorem~10.13 in the full arXiv version of O'Donnell and Wu
\cite[Theorem~10.13]{ODonnellWu2020Full}, reformulated in the unitary-variable
notation used here.

\begin{lemma}[O'Donnell--Wu spectral approximation]
\label{lem:odw}
Fix integers $K,D\ge1$ and positive rational parameters
$R,\eta\in\mathbb Q_{>0}$.  There exists an integer
$N_0=N_0(K,D,R,\eta)$ and one deterministic algorithm with the following
property.  For every integer $N\ge N_0$, on input $N$ the algorithm runs
in time polynomial in $N$ and outputs an integer
$N'=N'_{K,D,R,\eta}(N)$ together with permutations
$\sigma_1,\ldots,\sigma_K\in S_{N'}$.  The output sizes satisfy
\begin{equation}\label{eq:odw-output-size}
   N\le N'_{K,D,R,\eta}(N),
   \qquad
   N'_{K,D,R,\eta}(N)-N=o_{K,D,R,\eta}(N)
   \quad(N\to\infty).
\end{equation}

Let
\[
   U_i:=P_{\sigma_i}\big|_{\mathcal H_{N'}^0},
   \qquad 1\le i\le K,
\]
and let
$u_i:=\lambda(g_i)\in C_r^*(\F_K)\subseteq L(\F_K)$ be the canonical
free Haar unitaries.  The same finite tuple $(U_1,\ldots,U_K)$ satisfies,
simultaneously for every self-adjoint reduced noncommutative
$*$-polynomial
\[
   p=\sum_w c_w w
\]
in the $K$ unitary variables and their formal adjoints, whose reduced
degree is at most $D$ and whose reduced-word coefficients satisfy
$|c_w|\le R$, the estimate
\begin{equation}\label{eq:odw-hausdorff}
   d_{\mathrm H}\!\left(
      \spec p(U_1,\ldots,U_K),
      \spec p(u_1,\ldots,u_K)
   \right)
   \le\eta.
\end{equation}
\end{lemma}

We apply Lemma~\ref{lem:odw} only with
\begin{equation}\label{eq:odw-parameters-used}
   D=2,
   \qquad
   R=1.
\end{equation}
\section{A deterministic entropy criterion}\label{sec:deterministic-criterion}

Following Collins's argument
\cite{Collins2018}, this section separates the entropy mechanism from the
free-probabilistic approximation used later.  We begin with an arbitrary
finite tuple of unitaries.  The associated quadratic radius will be shown
to equal the exact Hilbert--Schmidt radius of the complementary output body
and to determine its maximum purity.  A Bell-state collision then supplies
a universal two-copy upper bound.  Their comparison yields a fully
finite-dimensional criterion containing no randomness or asymptotic limit.

\subsection{Mixed-unitary channels, branch Gram matrices, and the quadratic radius}

Fix integers $m\ge1$ and $K\ge2$, and let
$U_1,\ldots,U_K\in U(m)$.  Define
\begin{equation}\label{eq:stinespring}
  V:\C^m\longrightarrow\C^m\otimes\C^K,
  \qquad
  Vx:=\frac1{\sqrt K}\sum_{i=1}^K U_ix\otimes\ket{i}.
\end{equation}
Since
\[
   V^*V=\frac1K\sum_{i=1}^K U_i^*U_i=I_m,
\]
this map is an isometry.  Tracing out either tensor factor gives the
complementary pair
\begin{align}
  \mathcal E(X)
  &:=\Tr_{\C^K}(VXV^*)
    =\frac1K\sum_{i=1}^K U_iXU_i^*,
    \label{eq:mixed-unitary}\\
  \Phi(X)
  &:=\Tr_{\C^m}(VXV^*)
    =\frac1K\left[\Tr(U_iXU_j^*)\right]_{i,j=1}^K.
    \label{eq:complementary-channel}
\end{align}
Thus $\mathcal E:M_m(\C)\to M_m(\C)$ is a uniform mixed-unitary
channel, and $\Phi:M_m(\C)\to M_K(\C)$ is its complementary channel.

For $A=(a_{ij})\in M_K(\C)$, introduce the quadratic operator
\begin{equation}\label{eq:quadratic-operator}
   \mathcal Q_U(A)
   :=\sum_{i,j=1}^K a_{ij}U_i^*U_j\in M_m(\C).
\end{equation}
The channel formula gives the duality identity
\begin{equation}\label{eq:quadratic-pairing}
   \Tr\bigl(A\Phi(X)\bigr)
   =\frac1K\Tr\bigl(X\mathcal Q_U(A)\bigr),
   \qquad
   A\in M_K(\C),\quad X\in M_m(\C).
\end{equation}
Moreover, $\mathcal Q_U(A)$ is self-adjoint whenever $A$ is
self-adjoint.

For a unit vector $x\in\C^m$, define the branch-orbit Gram matrix, in the
index convention adapted to \eqref{eq:complementary-channel}, by
\begin{equation}\label{eq:gram-matrix}
   [G_x]_{ij}
   :=\bra{x}U_j^*U_i\ket{x},
   \qquad 1\le i,j\le K.
\end{equation}
Then $G_x\ge0$, every diagonal entry of $G_x$ equals one, and
\begin{equation}\label{eq:pure-output-gram}
   \Phi(\proj{x})=\frac{G_x}{K}.
\end{equation}
In fact, every output of $\Phi$ has diagonal $I_K/K$.

The quadratic radius is
\begin{equation}\label{eq:gamma}
  \Gamma_K(U_1,\ldots,U_K)
  :=\sup_{\substack{A=A^*\in M_K(\C),\ \Tr A=0\\
                     \|A\|_{\HS}=1}}
       \|\mathcal Q_U(A)\|,
\end{equation}
where $\|A\|_{\HS}:=(\Tr(A^*A))^{1/2}$.  By homogeneity,
\begin{equation}\label{eq:gamma-homogeneous}
   \|\mathcal Q_U(A)\|
   \le
   \Gamma_K(U_1,\ldots,U_K)\,\|A\|_{\HS}
\end{equation}
for every Hermitian traceless $A$. 

Only the off-diagonal part of a traceless matrix contributes.  Indeed,
\[
   \mathcal Q_U(A)
   =\sum_{i\ne j}a_{ij}U_i^*U_j+(\Tr A)I_m
   =\sum_{i\ne j}a_{ij}U_i^*U_j.
\]
Consequently,
\begin{equation}\label{eq:gamma-zero-diagonal}
  \Gamma_K(U_1,\ldots,U_K)
  =\sup_{\substack{A=A^*\in M_K(\C),\ a_{ii}=0\ (1\le i\le K)\\
                    \|A\|_{\HS}=1}}
       \|\mathcal Q_U(A)\|.
\end{equation}
To verify the reverse inequality implicit here, let $A^\circ$ be the
off-diagonal part of an admissible $A$.  Then
$\mathcal Q_U(A)=\mathcal Q_U(A^\circ)$ and
$\|A^\circ\|_{\HS}\le1$; if $A^\circ\ne0$, normalize it, while if
$A^\circ=0$, then $\mathcal Q_U(A)=0$.

\subsection{Exact Gram radius, maximum output purity, and one-copy entropy}

The next theorem gives an exact geometric interpretation of the
certificate in \eqref{eq:gamma}.

\begin{theorem}[Exact Gram radius and maximum output purity]
\label{thm:gram-purity}
For every finite unitary tuple $U=(U_1,\ldots,U_K)$ and the channel
$\Phi$ in \eqref{eq:complementary-channel},
\begin{equation}\label{eq:exact-gram-radius}
  \Gamma_K(U_1,\ldots,U_K)
  =\max_{\|x\|=1}\|G_x-I_K\|_{\HS}
  =K\max_{X\in\Dens(\C^m)}
    \left\|\Phi(X)-\frac{I_K}{K}\right\|_{\HS}.
\end{equation}
Equivalently,
\begin{equation}\label{eq:frame-potential}
  \Gamma_K(U_1,\ldots,U_K)^2
  =\max_{\|x\|=1}
    \sum_{i\ne j}
    \left|\bra{x}U_i^*U_j\ket{x}\right|^2.
\end{equation}
Moreover,
\begin{equation}\label{eq:exact-max-purity}
  \max_{X\in\Dens(\C^m)}\Tr\bigl(\Phi(X)^2\bigr)
  =\frac1K+
    \frac{\Gamma_K(U_1,\ldots,U_K)^2}{K^2}.
\end{equation}
Consequently,
\begin{equation}\label{eq:exact-h2}
  H_{\min,2}(\Phi)
  =-\log\left(
       \frac1K+
       \frac{\Gamma_K(U_1,\ldots,U_K)^2}{K^2}
     \right).
\end{equation}
\end{theorem}

\begin{proof}
For Hermitian traceless $A$, the operator $\mathcal Q_U(A)$ is
self-adjoint, and hence
\begin{align*}
  \|\mathcal Q_U(A)\|
  &=\max_{\|x\|=1}
    \left|\bra{x}\mathcal Q_U(A)\ket{x}\right|\\
  &=\max_{\|x\|=1}\left|\Tr(A G_x)\right|\\
  &=\max_{\|x\|=1}
    \left|\Tr\bigl(A(G_x-I_K)\bigr)\right|.
\end{align*}
The last identity uses $\Tr A=0$.  Taking the supremum over $A$,
interchanging the two suprema, and applying Hilbert--Schmidt duality on
the real Hilbert space of traceless Hermitian matrices gives
\[
   \Gamma_K(U_1,\ldots,U_K)
   =\max_{\|x\|=1}\|G_x-I_K\|_{\HS}.
\]
Here $G_x-I_K$ is itself Hermitian and traceless.  By
\eqref{eq:pure-output-gram}, this proves the second equality in
\eqref{eq:exact-gram-radius} for pure inputs.  If
$X=\sum_r t_r\proj{x_r}$ is a convex decomposition of a mixed input, then
\[
   \Phi(X)-\frac{I_K}{K}
   =\sum_r t_r\left(
      \Phi(\proj{x_r})-\frac{I_K}{K}
    \right),
\]
so convexity of the Hilbert--Schmidt norm shows that a mixed input cannot
increase the maximum.  Since pure inputs are included, the second equality
in \eqref{eq:exact-gram-radius} follows.

Because the diagonal entries of $G_x$ are one,
\[
   \Tr(G_x)=K,
   \qquad
   \Tr(G_x^2)=K+\|G_x-I_K\|_{\HS}^2.
\]
Thus
\[
   \Tr\left(\frac{G_x}{K}\right)^2
   =\frac1K+
    \frac{\|G_x-I_K\|_{\HS}^2}{K^2}.
\]
The map $\rho\mapsto\Tr(\rho^2)=\|\rho\|_{\HS}^2$ is convex, so its
maximum over the convex input state space is attained on a pure input.
Maximizing the preceding identity and using
\eqref{eq:exact-gram-radius} proves \eqref{eq:exact-max-purity}.
Since $-\log t$ is decreasing, \eqref{eq:exact-h2} follows immediately.
Finally, expanding $\|G_x-I_K\|_{\HS}^2$ and using its zero diagonal gives
\eqref{eq:frame-potential}.
\end{proof}

\begin{proposition}[Exact one-copy geometry and entropy bound]
\label{prop:one-copy}
If
\[
   \Gamma_K(U_1,\ldots,U_K)\le C,
\]
then every $X\in\Dens(\C^m)$ satisfies
\begin{align}
   \left\|\Phi(X)-\frac{I_K}{K}\right\|_{\HS}
   &\le\frac{C}{K},
   \label{eq:flattening}\\
   \Tr\bigl(\Phi(X)^2\bigr)
   &\le\frac1K+\frac{C^2}{K^2}.
   \label{eq:purity-bound}
\end{align}
Consequently,
\begin{equation}\label{eq:one-copy-entropy}
   \Hmin(\Phi)
   \ge
   \log K-\log\left(1+\frac{C^2}{K}\right)
   \ge
   \log K-\frac{C^2}{K}.
\end{equation}
\end{proposition}

\begin{proof}
The first two claims follow directly from
Theorem~\ref{thm:gram-purity}.  For every density matrix $\rho$, the von
Neumann entropy dominates the R\'enyi-$2$ entropy:
\[
   H(\rho)\ge-\log\Tr(\rho^2).
\]
Indeed, this is Jensen's inequality for the concave logarithm applied to
the eigenvalues of $\rho$, with the zero-eigenvalue case obtained by
continuity.  Combining this inequality with \eqref{eq:purity-bound} gives
\[
   H(\Phi(X))
   \ge
   -\log\left(\frac1K+\frac{C^2}{K^2}\right)
   =\log K-\log\left(1+\frac{C^2}{K}\right).
\]
Taking the minimum over $X$ proves the first entropy bound, and the second
follows from $\log(1+t)\le t$ for $t\ge0$.
\end{proof}

For each input $x$, the vectors $U_1x,\ldots,U_Kx$ form a $K$-tuple of unit vectors.  Equation~\eqref{eq:frame-potential} identifies $\Gamma_K^2$ with
the largest off-diagonal frame potential of a branch orbit.  Thus the
quadratic certificate measures the largest aggregate branch coherence
available to any input, rather than merely bounding that coherence from
above.

\subsection{The universal Bell-branch collision}

Complex conjugation below is taken entrywise in the fixed standard bases.
For a quantum channel $\Theta$, write
\[
   \overline\Theta(X):=\overline{\Theta(\overline X)}.
\]
Then $\overline\Theta$ is again a quantum channel and
$\Hmin(\overline\Theta)=\Hmin(\Theta)$.  For the channels in
\eqref{eq:mixed-unitary}--\eqref{eq:complementary-channel}, conjugation is
equivalent to replacing each $U_i$ by $\overline{U_i}$.

Let
\[
   \ket{\Omega_m}:=\frac1{\sqrt m}\sum_{a=1}^m\ket{a}\otimes\ket{a}
\]
be the normalized maximally entangled vector.  For every unitary
$U\in U(m)$,
\begin{equation}\label{eq:bell-invariance}
   (U\otimes\overline U)\ket{\Omega_m}=\ket{\Omega_m}.
\end{equation}
Consequently, when $\mathcal E\otimes\overline{\mathcal E}$ is evaluated
on $\proj{\Omega_m}$, the $K$ branches indexed by $(i,i)$ all produce the
same pure state.  The following is the standard Bell-state estimate used
in Hastings's random-unitary construction \cite{Hastings2009}; see also
\cite[Section~5.1]{FukudaKingMoser2010} and
\cite[Proposition~3.2]{Collins2018}.  We include the short entropy argument
for completeness.

\begin{proposition}[Universal two-copy estimate]\label{prop:two-copy}
For arbitrary unitaries $U_1,\ldots,U_K\in U(m)$, the complementary
channel $\Phi$ satisfies
\begin{equation}\label{eq:two-copy-conjugate}
  \Hmin(\Phi\otimes\overline\Phi)
  \le 2\log K-\frac{\log K}{K}.
\end{equation}
If every $U_i$ is real, then $\overline\Phi=\Phi$ and therefore
\begin{equation}\label{eq:two-copy-self}
  \Hmin(\Phi^{\otimes2})
  \le 2\log K-\frac{\log K}{K}.
\end{equation}
\end{proposition}

\begin{proof}
Evaluate $\mathcal E\otimes\overline{\mathcal E}$ on
$\proj{\Omega_m}$.  Its output is
\begin{equation}\label{eq:bell-mixture}
   (\mathcal E\otimes\overline{\mathcal E})(\proj{\Omega_m})
   =\frac1{K^2}\sum_{i,j=1}^K\proj{\psi_{ij}},
   \qquad
   \ket{\psi_{ij}}
   :=(U_i\otimes\overline{U_j})\ket{\Omega_m}.
\end{equation}
By \eqref{eq:bell-invariance},
$\ket{\psi_{ii}}=\ket{\Omega_m}$ for every $i$.  Grouping these $K$ identical terms gives an ensemble representation with one weight
$1/K$ and $K^2-K$ further weights equal to $1/K^2$.

For completeness, suppose
$\rho=\sum_\alpha q_\alpha\proj{\varphi_\alpha}$ with
$\|\varphi_\alpha\|=1$, and define a map $T$ on the standard basis by
$T e_\alpha:=\sqrt{q_\alpha}\,\varphi_\alpha$.  Then
$TT^*=\rho$, while the Gram matrix $G:=T^*T$ has diagonal
$(q_\alpha)_\alpha$ and the same nonzero eigenvalues as $\rho$.
The Schur--Horn theorem says that the diagonal vector of $G$ is majorized
by its eigenvalue vector; Schur-concavity of Shannon entropy therefore
gives $H(\rho)\le-\sum_\alpha q_\alpha\log q_\alpha$.  Applying this to the preceding ensemble gives
\begin{align*}
  H\bigl((\mathcal E\otimes\overline{\mathcal E})
      (\proj{\Omega_m})\bigr)
  &\le -\frac1K\log\frac1K
       -(K^2-K)\frac1{K^2}\log\frac1{K^2}\\
  &=2\log K-\frac{\log K}{K}.
\end{align*}

The vector obtained by applying $V\otimes\overline V$ to
$\ket{\Omega_m}$ is pure.  Up to the canonical reordering of tensor
factors, its two reduced states are
\[
   (\mathcal E\otimes\overline{\mathcal E})(\proj{\Omega_m})
   \quad\text{and}\quad
   (\Phi\otimes\overline\Phi)(\proj{\Omega_m}).
\]
They therefore have the same nonzero eigenvalues and the same entropy.
Since minimum output entropy is bounded above by the entropy at this input,
\eqref{eq:two-copy-conjugate} follows.  If the matrices $U_i$ are real,
then $\overline\Phi=\Phi$, proving \eqref{eq:two-copy-self}.
\end{proof}

\subsection{The deterministic criterion}

Combining the exact one-copy lower bound with the Bell-state upper bound
yields the following finite-dimensional criterion.

\begin{corollary}[Deterministic criterion]
\label{cor:deterministic-criterion}
Suppose $U_1,\ldots,U_K\in U(m)$ and
$\Gamma_K(U_1,\ldots,U_K)\le C$.  Then
\begin{align}
  2\Hmin(\Phi)-\Hmin(\Phi\otimes\overline\Phi)
  &\ge
  \frac{\log K}{K}
  -2\log\left(1+\frac{C^2}{K}\right)
  \label{eq:criterion-gap-conjugate-strong}\\
  &\ge\frac{\log K-2C^2}{K}.
  \label{eq:criterion-gap-conjugate}
\end{align}
If, in addition, $U_1,\ldots,U_K\in O(m)$, then
\begin{align}
  2\Hmin(\Phi)-\Hmin(\Phi^{\otimes2})
  &\ge
  \frac{\log K}{K}
  -2\log\left(1+\frac{C^2}{K}\right)
  \label{eq:criterion-gap-strong}\\
  &\ge\frac{\log K-2C^2}{K}.
  \label{eq:criterion-gap}
\end{align}
In particular, $\log K>2C^2$ is a sufficient condition for a strict
product--conjugate violation in general and a strict tensor-square
violation for a real orthogonal tuple.
\end{corollary}

\begin{proof}
By Proposition~\ref{prop:one-copy},
\[
   2\Hmin(\Phi)
   \ge
   2\log K-2\log\left(1+\frac{C^2}{K}\right).
\]
Subtracting \eqref{eq:two-copy-conjugate} proves
\eqref{eq:criterion-gap-conjugate-strong}.  The second line follows from
$\log(1+t)\le t$.  In the real orthogonal case,
$\overline\Phi=\Phi$, giving \eqref{eq:criterion-gap-strong} and
\eqref{eq:criterion-gap}.
\end{proof}

\section{The Haagerup--O'Donnell--Wu construction}\label{sec:construction}

We now construct the finite orthogonal tuple required by
Corollary~\ref{cor:deterministic-criterion}. First, every admissible matrix \(A\) in \eqref{eq:gamma} determines a self-adjoint reduced scalar polynomial lying in one fixed
O'Donnell--Wu class.  Second, Haagerup's inequality bounds the corresponding free-group operator uniformly by \(3\).  The simultaneous spectral approximation in Lemma~\ref{lem:odw} then transfers this bound to one deterministically constructed finite permutation tuple.

\subsection{Quadratic tests and the O'Donnell--Wu class}

Fix \(K\ge2\), and let \(Z_1,\ldots,Z_K\) be the formal unitary variables
introduced in Section~\ref{sec:preliminaries}.  For every Hermitian traceless
matrix
\[
   A=(a_{ij})=A^*\in M_K(\C),
   \qquad
   \Tr A=0,
\]
define
\begin{equation}\label{eq:pA}
   p_A(Z_1,\ldots,Z_K)
   :=
   \sum_{i,j=1}^K a_{ij}Z_i^*Z_j.
\end{equation}
This polynomial is chosen precisely to encode the quadratic test appearing
in \eqref{eq:gamma}: for every unitary tuple
\(W_1,\ldots,W_K\),
\begin{equation}\label{eq:pA-quadratic-operator}
   p_A(W_1,\ldots,W_K)
   =
   \sum_{i,j=1}^K a_{ij}W_i^*W_j
   =
   \mathcal Q_W(A).
\end{equation}

Under unitary reduction, each diagonal monomial \(Z_i^*Z_i\) becomes the
identity.  Since
\(\sum_i a_{ii}=\Tr A=0\), the resulting constant term cancels, and hence
\begin{equation}\label{eq:pA-reduced}
   p_A
   =
   \sum_{i\ne j}a_{ij}Z_i^*Z_j.
\end{equation}
We now verify that, under the normalization
\(\|A\|_{\HS}=1\) used in \eqref{eq:gamma}, these polynomials belong to one fixed scalar O'Donnell--Wu class:
\begin{itemize}
\item
\emph{Scalar coefficients.}
The coefficients \(a_{ij}\) are complex scalars, so this is the scalar
specialization of Lemma~\ref{lem:odw}.

\item
\emph{Self-adjointness.}
Since \(A=A^*\), we have \(a_{ji}=\overline{a_{ij}}\), and therefore
\[
   p_A^*
   =
   \sum_{i,j=1}^K \overline{a_{ij}}\,Z_j^*Z_i
   =
   p_A.
\]

\item
\emph{Reduced degree.}
For \(i\ne j\), the word \(Z_i^*Z_j\) is already reduced and has reduced
degree two.  Under the correspondence \(Z_i\leftrightarrow g_i\), it
becomes \(g_i^{-1}g_j\).  These words are pairwise distinct as
\((i,j)\) ranges over \(i\ne j\).  Thus \(p_A\) is a reduced polynomial
of degree at most \(D=2\).

\item
\emph{Coefficient bound.}
For every \(i,j\),
\[
   |a_{ij}|
   \le
   \left(\sum_{r,s=1}^K |a_{rs}|^2\right)^{1/2}
   =
   \|A\|_{\HS}
   =
   1.
\]
Hence, the reduced-word coefficients satisfy the uniform bound \(R=1\).
\end{itemize}

Consequently, the entire family of polynomials \(p_A\), with \(A\)
admissible in \eqref{eq:gamma}, lies in the single scalar
O'Donnell--Wu class with
\[
   D=2,
   \qquad
   R=1.
\]
Since Lemma~\ref{lem:odw} is simultaneous over this whole polynomial
class, one and the same finite permutation tuple provides the required
spectral approximation for every admissible \(A\).

\subsection{The free-group quadratic estimate}

Under the free-group evaluation
\(
Z_i\mapsto u_i=\lambda(g_i)
\),
each polynomial \(p_A\) becomes a convolution operator supported on
reduced words of length two.  The length-two case of Haagerup's
inequality therefore gives a uniform operator-norm bound for the entire
family.

\begin{lemma}[Uniform free quadratic estimate]
\label{lem:free-quadratic-estimate}
Let
\[
   u_i=\lambda(g_i)\in C_r^*(\F_K)\subseteq L(\F_K),
   \qquad 1\le i\le K,
\]
be the canonical free Haar unitaries arising from the left regular
representation of \(\F_K\).  For every Hermitian traceless
\(A=(a_{ij})\in M_K(\C)\),
\begin{equation}\label{eq:haagerup-pA-general}
   \|p_A(u_1,\ldots,u_K)\|
   \le
   3\left(\sum_{i\ne j}|a_{ij}|^2\right)^{1/2}.
\end{equation}
In particular, if \(\|A\|_{\HS}=1\), then
\begin{equation}\label{eq:haagerup-pA}
   \|p_A(u_1,\ldots,u_K)\|\le3.
\end{equation}
\end{lemma}

\begin{proof}
By \eqref{eq:pA-reduced} and
\(u_i^*u_j=\lambda(g_i^{-1}g_j)\),
\[
   p_A(u_1,\ldots,u_K)
   =
   \sum_{i\ne j}a_{ij}\lambda(g_i^{-1}g_j).
\]
Define the finitely supported coefficient function
\(f_A:\F_K\to\C\) by
\[
   f_A(g)
   :=
   \begin{cases}
      a_{ij},
      & g=g_i^{-1}g_j\text{ for some }i\ne j,\\
      0,
      & \text{otherwise}.
   \end{cases}
\]
This is well defined because the reduced words
\(g_i^{-1}g_j\), \(i\ne j\), are pairwise distinct.
Moreover,
\[
   \supp(f_A)\subseteq S_2,
   \qquad
   \|f_A\|_{\ell^2(\F_K)}^2
   =
   \sum_{i\ne j}|a_{ij}|^2.
\]

Recall that \(\lambda(g)\) denotes the left-regular unitary associated
with a group element \(g\in\F_K\), whereas, by linear extension,
\[
   \lambda(f_A)
   :=
   \sum_{g\in\F_K}f_A(g)\lambda(g)
\]
denotes the corresponding left convolution operator on
\(\ell^2(\F_K)\).  With this notation,
\[
   p_A(u_1,\ldots,u_K)=\lambda(f_A).
\]
Since \(f_A\) is supported on the sphere \(S_2\), the \(d=2\) case of
Lemma~\ref{lem:haagerup} gives
\[
   \|p_A(u_1,\ldots,u_K)\|
   =
   \|\lambda(f_A)\|
   \le
   3\|f_A\|_{\ell^2(\F_K)}
   =
   3\left(\sum_{i\ne j}|a_{ij}|^2\right)^{1/2}.
\]
This proves \eqref{eq:haagerup-pA-general}.  If
\(\|A\|_{\HS}=1\), then
\[
   \sum_{i\ne j}|a_{ij}|^2
   \le
   \sum_{i,j}|a_{ij}|^2
   =
   \|A\|_{\HS}^2
   =
   1,
\]
which yields \eqref{eq:haagerup-pA}.
\end{proof}

\subsection{Deterministic transfer to a finite orthogonal tuple}

We now use the simultaneous spectral approximation of
Lemma~\ref{lem:odw} to transfer the uniform free-group estimate in
Lemma~\ref{lem:free-quadratic-estimate} to finite dimension.

\begin{proposition}[Finite near-free standard representations]
\label{prop:near-free-tuple}
Fix $K\ge2$ and $\eta\in\mathbb Q_{>0}$.  There exists an integer
$N_0=N_0(K,\eta)$ and one deterministic algorithm
$\mathcal A_{K,\eta}$ with the following property.  For every integer
$N\ge N_0$, on input $N$ the algorithm runs in time polynomial in $N$ and
outputs an integer $N'=N'_{K,\eta}(N)$ and permutations
\[
   \sigma_1,\ldots,\sigma_K\in S_{N'}
\]
such that
\begin{equation}\label{eq:near-free-output-size}
   N\le N',
   \qquad
   N'_{K,\eta}(N)-N=o_{K,\eta}(N)
   \quad(N\to\infty).
\end{equation}
For any real orthonormal identification
$Q:\C^{N'-1}\to\mathcal H_{N'}^0$, the matrices
\[
   U_i:=Q^*P_{\sigma_i}Q\in O(N'-1),
   \qquad 1\le i\le K,
\]
satisfy
\begin{equation}\label{eq:gamma-bound}
   \Gamma_K(U_1,\ldots,U_K)\le3+\eta.
\end{equation}
\end{proposition}

\begin{proof}
Apply Lemma~\ref{lem:odw} with $D=2$ and $R=1$.  For the fixed
parameters $K$ and $\eta$, denote the resulting deterministic algorithm
by $\mathcal A_{K,\eta}$.  For every sufficiently large $N$, it outputs
an integer $N'$ and permutations
$\sigma_1,\ldots,\sigma_K\in S_{N'}$ satisfying
\eqref{eq:odw-output-size}.  This gives
\eqref{eq:near-free-output-size} after increasing $N_0$ if necessary.

Let
$\widetilde U_i:=P_{\sigma_i}|_{\mathcal H_{N'}^0}$ and choose any real
orthonormal identification
$Q:\C^{N'-1}\to\mathcal H_{N'}^0$.  Put
$U_i:=Q^*\widetilde U_i Q$.  Each $U_i$ is real orthogonal.  Now take
\[
   A=A^*\in M_K(\C),
   \qquad
   \Tr A=0,
   \qquad
   \|A\|_{\HS}=1.
\]
By the preceding reduction, the associated polynomial $p_A$ is a
self-adjoint reduced scalar polynomial of degree at most two, and every
reduced-word coefficient has modulus at most one.  The simultaneous
conclusion of Lemma~\ref{lem:odw} therefore gives
\[
   d_{\mathrm H}\!\left(
      \spec p_A(\widetilde U_1,\ldots,\widetilde U_K),
      \spec p_A(u_1,\ldots,u_K)
   \right)
   \le\eta
\]
for every admissible $A$.  Conjugation by $Q$ does not change spectra, so
the same estimate holds with the matrices $U_i$ in place of
$\widetilde U_i$.

Both polynomial evaluations are self-adjoint.  Hence
\eqref{eq:hausdorff-to-norm}, \eqref{eq:pA-quadratic-operator}, and
Lemma~\ref{lem:free-quadratic-estimate} imply
\begin{align*}
   \|\mathcal Q_U(A)\|
   &=\|p_A(U_1,\ldots,U_K)\|\\
   &\le\|p_A(u_1,\ldots,u_K)\|+\eta\\
   &\le3+\eta.
\end{align*}
Taking the supremum over all admissible $A$ proves
\eqref{eq:gamma-bound}.
\end{proof}

\subsection{Proof of main results}

We now combine the finite orthogonal tuple constructed above with the
deterministic entropy criterion to complete the proof of the main theorem.

\begin{proof}[Proof of Theorem~\ref{thm:main}]
Fix $K\ge2$ and $\eta\in\mathbb Q_{>0}$ satisfying
\eqref{eq:main-threshold}.  Proposition~\ref{prop:near-free-tuple}
provides an integer $N_0$ and one deterministic algorithm
$\mathcal A_{K,\eta}$.  For every $N\ge N_0$, on input $N$ this algorithm
outputs an integer $N'=N'_{K,\eta}(N)$ and permutations
$\sigma_1,\ldots,\sigma_K\in S_{N'}$ with the size and running-time
properties stated in Theorem~\ref{thm:main}.

Choose any real orthonormal identification
$Q:\C^{N'-1}\to\mathcal H_{N'}^0$ and let
$U_i=Q^*P_{\sigma_i}Q$.  By Proposition~\ref{prop:near-free-tuple},
\[
   \Gamma_K(U_1,\ldots,U_K)\le3+\eta.
\]
Let $\Phi_{N'}^Q:M_{N'-1}(\C)\to M_K(\C)$ be the complementary channel
in \eqref{eq:channel-main-results}.  It admits a real Stinespring
realization, and Remark~\ref{rem:channel-construction} shows that its
input-orthogonal equivalence class is determined canonically by the output
permutations.  Applying Corollary~\ref{cor:deterministic-criterion} with
$C=3+\eta$ yields
\begin{align*}
   2\Hmin(\Phi_{N'}^Q)
   -\Hmin\bigl((\Phi_{N'}^Q)^{\otimes2}\bigr)
   &\ge
   \frac{\log K}{K}
   -2\log\left(1+\frac{(3+\eta)^2}{K}\right)\\
   &\ge
   \frac{\log K-2(3+\eta)^2}{K}.
\end{align*}
The final quantity is strictly positive by
\eqref{eq:main-threshold}.  This proves
\eqref{eq:main-gap-strong}, \eqref{eq:main-gap}, and all remaining
assertions of Theorem~\ref{thm:main}.
\end{proof}

The concrete choice in Corollary~\ref{cor:concrete} now requires only an
explicit verification of the threshold \eqref{eq:main-threshold} and of
the resulting entropy gap.

\begin{proof}[Proof of Corollary~\ref{cor:concrete}]
Take
\[
   K=2^{26},
   \qquad
   \eta=10^{-3}.
\]
Then
\[
   \log K=26\log2,
   \qquad
   2(3+\eta)^2
   =
   2\left(\frac{3001}{1000}\right)^2
   =
   \frac{9006001}{500000}.
\]
Using the convergent expansion
\[
   \log2
   =
   2\operatorname{artanh}(1/3)
   =
   2\sum_{n=0}^{\infty}
     \frac{1}{(2n+1)3^{2n+1}},
\]
whose terms are positive, we obtain the elementary rational lower bound
\[
   \log2
   >
   2\left(
      \frac13+\frac1{81}+\frac1{1215}
   \right)
   =
   \frac{842}{1215}.
\]
Consequently,
\begin{align*}
   \log K-2(3+\eta)^2
   &=
   26\log2-2(3.001)^2 \\
   &>
   \frac{21892}{1215}
   -
   \frac{9006001}{500000} \\
   &=
   \frac{741757}{121500000}
   >0.
\end{align*}
Thus \eqref{eq:main-threshold} holds.  Theorem~\ref{thm:main} therefore
applies, and its entropy-gap estimate gives
\[
   2\Hmin(\Phi)-\Hmin(\Phi^{\otimes2})
   >
   \frac{741757}
        {121500000\cdot2^{26}}
   >
   9.09\times10^{-11}\ \text{nats}.
\]
This proves Corollary~\ref{cor:concrete}.
\end{proof}

\section{A standard covariant extension and Holevo superadditivity}
\label{sec:holevo-conversion}

The passage from minimum-output-entropy nonadditivity to Holevo superadditivity is classically implemented by covariant or unital
extensions; it is one of the standard devices underlying the equivalence
of additivity questions \cite{Shor2004}, and related extensions are used in
detailed accounts of Hastings's construction
\cite{FukudaKingMoser2010}.  We record a self-contained specialization adapted to the real self-tensor channels constructed here.  The point is
not a new general equivalence theorem, but an exact operational conversion for the present deterministic family.

For a channel $\Theta$, its one-shot Holevo quantity is
\begin{equation}\label{eq:holevo-definition}
  \chi(\Theta)
  :=\sup_{\{p_a,\rho_a\}}
  \left[
    H\!\left(\sum_a p_a\Theta(\rho_a)\right)
    -\sum_a p_a H(\Theta(\rho_a))
  \right].
\end{equation}

\begin{proposition}[Flagged one-design extension]
\label{prop:flagged-holevo}
Let $\Phi:\B(\mathcal H)\to M_K(\C)$ be a quantum channel, and let
$\{W_g:g\in\mathcal G\}\subset U(K)$ be a finite unitary one-design, in
the sense that
\begin{equation}\label{eq:one-design}
  \frac1{|\mathcal G|}\sum_{g\in\mathcal G}W_gYW_g^*
  =\frac{\Tr(Y)}{K}I_K
  \qquad(Y\in M_K(\C)).
\end{equation}
For $X\in\B(\C^{|\mathcal G|}\otimes\mathcal H)$, write
\[
   X_{gg}:=(\bra g\otimes I)X(\ket g\otimes I),
\]
and define
\begin{equation}\label{eq:flagged-channel}
  \widehat\Phi(X)
  :=\sum_{g\in\mathcal G}W_g\Phi(X_{gg})W_g^*.
\end{equation}
Then $\widehat\Phi$ is a quantum channel and, for every integer $n\ge1$,
\begin{align}
  \Hmin(\widehat\Phi^{\otimes n})
  &=\Hmin(\Phi^{\otimes n}),
  \label{eq:flagged-hmin}\\
  \chi(\widehat\Phi^{\otimes n})
  &=n\log K-\Hmin(\Phi^{\otimes n}).
  \label{eq:flagged-chi}
\end{align}
In particular,
\begin{equation}\label{eq:gap-preservation}
  \chi(\widehat\Phi^{\otimes2})-2\chi(\widehat\Phi)
  =2\Hmin(\Phi)-\Hmin(\Phi^{\otimes2}).
\end{equation}
\end{proposition}

\begin{proof}
The maps $X\mapsto X_{gg}$ are completely positive, so
\eqref{eq:flagged-channel} is completely positive.  Moreover,
\[
   \Tr\widehat\Phi(X)
   =\sum_g\Tr X_{gg}
   =\Tr X,
\]
and hence $\widehat\Phi$ is trace preserving.

Fix $n\ge1$, write $\mathbf g=(g_1,\ldots,g_n)\in\mathcal G^n$, and set
$W_{\mathbf g}:=W_{g_1}\otimes\cdots\otimes W_{g_n}$.  For an input
state $X$ of $\widehat\Phi^{\otimes n}$, let
$X_{\mathbf g\mathbf g}$ be its diagonal flag blocks,
$t_{\mathbf g}:=\Tr X_{\mathbf g\mathbf g}$, and, whenever
$t_{\mathbf g}>0$, let
$\rho_{\mathbf g}:=X_{\mathbf g\mathbf g}/t_{\mathbf g}$.  Then
\[
   \widehat\Phi^{\otimes n}(X)
   =\sum_{\mathbf g}t_{\mathbf g}
      W_{\mathbf g}\Phi^{\otimes n}(\rho_{\mathbf g})W_{\mathbf g}^*.
\]
Concavity and unitary invariance of entropy imply
\[
   H\bigl(\widehat\Phi^{\otimes n}(X)\bigr)
   \ge\sum_{\mathbf g}t_{\mathbf g}
       H\bigl(\Phi^{\otimes n}(\rho_{\mathbf g})\bigr)
   \ge\Hmin(\Phi^{\otimes n}).
\]
Equality is attained by placing a minimizing input of
$\Phi^{\otimes n}$ in one flag block.  This proves
\eqref{eq:flagged-hmin}.

Every output of $\widehat\Phi^{\otimes n}$ acts on a space of dimension
$K^n$, and every individual output entropy is at least
$\Hmin(\widehat\Phi^{\otimes n})$.  Therefore
\[
   \chi(\widehat\Phi^{\otimes n})
   \le n\log K-\Hmin(\Phi^{\otimes n}).
\]
Let $\rho_\star$ minimize the output entropy of $\Phi^{\otimes n}$ and
use the equally likely input ensemble
\[
   \bigl\{\proj{\mathbf g}\otimes\rho_\star:
          \mathbf g\in\mathcal G^n\bigr\}.
\]
Its outputs are unitary conjugates of
$\Phi^{\otimes n}(\rho_\star)$ and hence all have entropy
$\Hmin(\Phi^{\otimes n})$.  The product one-design obtained from
\eqref{eq:one-design} makes their average equal to $I_{K^n}/K^n$.
Thus this ensemble attains the preceding upper bound, proving
\eqref{eq:flagged-chi}.  Equation~\eqref{eq:gap-preservation} follows by
subtracting the cases $n=1$ and $n=2$.
\end{proof}

\begin{proof}[Proof of Corollary~\ref{cor:holevo-main}]
Let $K=2^q$ and define the real single-qubit matrices
\[
  \mathsf X:=\begin{pmatrix}0&1\\1&0\end{pmatrix},
  \qquad
  \mathsf Z:=\begin{pmatrix}1&0\\0&-1\end{pmatrix}.
\]
For $a,b\in\{0,1\}^q$, set
\begin{equation}\label{eq:real-pauli-design}
  W_{a,b}
  :=\bigotimes_{r=1}^q\mathsf X^{a_r}\mathsf Z^{b_r}.
\end{equation}
These $K^2$ matrices are real orthogonal.  They form a unitary one-design:
the matrices $W_{c,d}$ form an orthogonal basis of $M_K(\C)$, and
conjugation by $W_{a,b}$ multiplies $W_{c,d}$ by
$(-1)^{a\cdot d+b\cdot c}$.  Averaging over $a,b$ annihilates every
nonidentity basis element and fixes the identity, which proves
\eqref{eq:one-design}.

Apply Proposition~\ref{prop:flagged-holevo} to $\Phi_{N'}$ and this
design.  If
\[
   V:\C^{N'-1}\longrightarrow\mathcal H_E\otimes\C^K
\]
is a real Stinespring isometry for $\Phi_{N'}$, where
$\mathcal H_E$ denotes the environment Hilbert space, then
\[
   \widehat V(\ket{a,b}\otimes x)
   :=
   \ket{a,b}\otimes
   (I_{\mathcal H_E}\otimes W_{a,b})Vx
\]
is a real Stinespring isometry for $\widehat\Phi_{N'}$.  Its input
dimension is $K^2(N'-1)$ and its output dimension is $K$.
Equation~\eqref{eq:gap-preservation}, followed by
Theorem~\ref{thm:main}, gives
\eqref{eq:main-holevo-equality} and \eqref{eq:main-holevo-gap}.
Finally, the Holevo--Schumacher--Westmoreland regularization formula
\cite{Holevo1998,SchumacherWestmoreland1997} gives
\[
   C(\widehat\Phi_{N'})
   =
   \lim_{n\to\infty}
   \frac1n\chi(\widehat\Phi_{N'}^{\otimes n})
   \ge
   \frac12\chi(\widehat\Phi_{N'}^{\otimes2})
   >
   \chi(\widehat\Phi_{N'}),
\]
which proves the capacity assertion.
\end{proof}

\section{Discussion and outlook}\label{sec:discussion}

The construction is a task-specific derandomization of the Collins
mixed-unitary mechanism.  It does not reproduce Haar randomness in all
moments or observables.  Instead, it isolates the scalar quadratic tests
actually used by the one-copy entropy argument and transfers precisely this
bounded class to deterministic permutation data.  The exact Gram identity
from Theorem~\ref{thm:gram-purity} sharpens that interpretation:
$\Gamma_K$ is the worst aggregate coherence of a branch orbit and exactly
determines the maximum output purity, while the Bell estimate supplies a
logically separate two-copy resource through matching-branch collisions.
The real standard representation turns this product--conjugate mechanism
into a self-tensor one.

\subsection{Asymptotic sharpness and finite near-saturation}

The upper bound $3$ in the free quadratic estimate
\eqref{eq:haagerup-pA} could a priori be loose because the coefficient
matrices are Hermitian and have zero diagonal.  We first show that these
restrictions do not improve the constant asymptotically.  Define
\begin{equation}\label{eq:structured-constant}
   C_K^{\mathrm{quad}}
   :=
   \sup_{\substack{A=A^*\in M_K(\C),\ a_{ii}=0\\
                    \|A\|_{\HS}=1}}
   \left\|
      \sum_{i\ne j}a_{ij}\lambda(g_i^{-1}g_j)
   \right\|.
\end{equation}
Haagerup's inequality gives $C_K^{\mathrm{quad}}\le3$.

\begin{proposition}[Asymptotic sharpness]
\label{prop:sharpness}
For every $K\ge2$,
\begin{equation}\label{eq:sharpness-lower}
   C_K^{\mathrm{quad}}
   \ge
   c_K
   :=\frac{3K-4}{\sqrt{K(K-1)}}.
\end{equation}
Consequently,
\begin{equation}\label{eq:sharpness-limit}
   \lim_{K\to\infty}C_K^{\mathrm{quad}}=3.
\end{equation}
\end{proposition}

\begin{proof}
Let $J_K$ denote the all-ones matrix and set
\[
   A_K:=\frac{J_K-I_K}{\sqrt{K(K-1)}}.
\]
Then $A_K=A_K^*$, its diagonal vanishes, and
$\|A_K\|_{\HS}=1$.  Put
$S_K:=\sum_{i=1}^K\lambda(g_i)$.  Since
\[
   S_K^*S_K
   =KI+\sum_{i\ne j}\lambda(g_i^{-1}g_j),
\]
we obtain
\begin{equation}\label{eq:sharpness-operator}
   \sum_{i\ne j}(A_K)_{ij}\lambda(g_i^{-1}g_j)
   =\frac{S_K^*S_K-KI}{\sqrt{K(K-1)}}.
\end{equation}
Akemann and Ostrand's norm formula gives
$\|S_K\|=2\sqrt{K-1}$ for $K\ge2$
\cite[Theorem~IV.G]{AkemannOstrand1976}.  Since $S_K^*S_K$ is positive,
$4(K-1)=\|S_K\|^2$ belongs to its spectrum.  Hence the self-adjoint
operator in \eqref{eq:sharpness-operator} has the spectral value
\[
   \frac{4(K-1)-K}{\sqrt{K(K-1)}}
   =c_K,
\]
and its norm is at least $c_K$.  Combining this with
$C_K^{\mathrm{quad}}\le3$ and letting $K\to\infty$ proves the claim.
\end{proof}

The two-sided spectral conclusion of the O'Donnell--Wu theorem transfers
this lower bound back to the actual finite tuple.

\begin{corollary}[Finite near-saturation]
\label{cor:finite-saturation}
Let $U_1,\ldots,U_K$ be a tuple obtained from
Lemma~\ref{lem:odw} with $D=2$, $R=1$, and spectral tolerance $\eta$.
Then
\begin{equation}\label{eq:finite-saturation}
   c_K-\eta
   \le\Gamma_K(U_1,\ldots,U_K)
   \le3+\eta.
\end{equation}
\end{corollary}

\begin{proof}
The upper bound is Proposition~\ref{prop:near-free-tuple}.  Apply
\eqref{eq:odw-hausdorff} to the admissible polynomial $p_{A_K}$.
By Proposition~\ref{prop:sharpness}, its free evaluation has operator norm
at least $c_K$.  Since both evaluations are self-adjoint,
\eqref{eq:hausdorff-to-norm} gives
\[
   \|p_{A_K}(U_1,\ldots,U_K)\|
   \ge
   \|p_{A_K}(u_1,\ldots,u_K)\|-\eta
   \ge c_K-\eta.
\]
Because $A_K$ is admissible in \eqref{eq:gamma}, the left-hand side is at
most $\Gamma_K(U_1,\ldots,U_K)$.
\end{proof}

For $K=2^{26}$ and $\eta=10^{-3}$, the preceding interval is
\[
   2.9989999627\ldots
   \le\Gamma_K(U_1,\ldots,U_K)
   \le3.001.
\]
Thus the radius estimate used by the construction is genuinely close to
saturation; the deterministic tuple cannot be expected to gain a
substantial improvement merely by being much better than its free target
on this test class.

\subsection{Reachable directions attain the full radius}

One might next try to optimize only over coefficient directions generated
by actual channel outputs.  The exact Gram geometry shows that this does
not reduce the radius.

\begin{proposition}[Reachable-direction obstruction]
\label{prop:reachable-no-go}
Assume $\Gamma_K(U_1,\ldots,U_K)>0$.  Then
\begin{equation}\label{eq:reachable-no-go}
  \sup_{\substack{X\in\Dens(\C^m)\\
                  \Phi(X)\ne I_K/K}}
  \left\|
    \mathcal Q_U\!\left(
      \frac{\Phi(X)-I_K/K}
           {\|\Phi(X)-I_K/K\|_{\HS}}
    \right)
  \right\|
  =\Gamma_K(U_1,\ldots,U_K).
\end{equation}
The same equality holds when the supremum is restricted to pure inputs.
\end{proposition}

\begin{proof}
Every normalized output direction in \eqref{eq:reachable-no-go} is
Hermitian, traceless, and Hilbert--Schmidt normalized; in fact, it also has
zero diagonal.  Hence the left-hand side is at most $\Gamma_K$.

By compactness and Theorem~\ref{thm:gram-purity}, there is a unit vector
$x_\star$ such that
\[
   \|G_{x_\star}-I_K\|_{\HS}
   =\Gamma_K(U_1,\ldots,U_K).
\]
Set
\[
  A_\star
  :=\frac{G_{x_\star}-I_K}
          {\|G_{x_\star}-I_K\|_{\HS}}.
\]
Equation~\eqref{eq:pure-output-gram} shows that
\[
   A_\star
   =\frac{\Phi(\proj{x_\star})-I_K/K}
          {\|\Phi(\proj{x_\star})-I_K/K\|_{\HS}},
\]
so $A_\star$ is a direction generated by a pure output.  Moreover,
\begin{align*}
  \bra{x_\star}\mathcal Q_U(A_\star)\ket{x_\star}
  &=\Tr(A_\star G_{x_\star})\\
  &=\frac{\Tr[(G_{x_\star}-I_K)G_{x_\star}]}
          {\|G_{x_\star}-I_K\|_{\HS}}\\
  &=\|G_{x_\star}-I_K\|_{\HS}\\
  &=\Gamma_K(U_1,\ldots,U_K),
\end{align*}
where $\Tr(G_{x_\star}-I_K)=0$ was used in the third line.  Therefore
$\|\mathcal Q_U(A_\star)\|\ge\Gamma_K$, proving the reverse inequality
and the pure-input assertion.
\end{proof}

Thus Hermitianity, zero diagonal, and restriction to directions reached by
actual outputs cannot improve the scalar radius certificate.  Quantitative
progress within this model must retain information discarded by the single
number $\Gamma_K$, change the one-copy certificate, or strengthen the
two-copy collision.

\subsection{The full output body and its entropy program}

The natural finer object is the full support function of the output body.
For a finite tuple $U=(U_1,\ldots,U_K)$, set
\[
   \mathcal K_U:=\Phi(\Dens(\C^m))\subset\Dens(\C^K).
\]
For every $A=A^*\in M_K(\C)$, the pairing identity gives
\begin{equation}\label{eq:finite-support-function}
  h_U(A)
  :=\max_{\rho\in\mathcal K_U}\Tr(A\rho)
  =\frac1K\lambda_{\max}\bigl(\mathcal Q_U(A)\bigr).
\end{equation}

There is an analogous canonical free body.  Let
$u_i=\lambda(g_i)\in C_r^*(\F_K)$ and define
\begin{equation}\label{eq:free-ucp-map}
   \mathcal T_\infty(A)
   :=\frac1K\sum_{i,j=1}^K a_{ij}u_i^*u_j,
   \qquad A=(a_{ij})\in M_K(\C).
\end{equation}
If
$v:=K^{-1/2}(u_1,\ldots,u_K)^T$, then
$\mathcal T_\infty(A)=v^*(A\otimes1)v$ and $v^*v=1$; hence
$\mathcal T_\infty$ is unital and completely positive.  For every state
$\omega$ on $C_r^*(\F_K)$, let $\rho_\omega\in\Dens(\C^K)$ be the unique
density matrix satisfying
\[
   \Tr(A\rho_\omega)=\omega(\mathcal T_\infty(A))
   \qquad(A\in M_K(\C)).
\]
Equivalently, its matrix entries are
\begin{equation}\label{eq:free-output-state-entries}
   [\rho_\omega]_{ij}
   =\frac1K\,\omega(u_j^*u_i),
   \qquad 1\le i,j\le K.
\end{equation}
Define
\begin{equation}\label{eq:free-output-body}
   \mathcal K_K^{\mathrm{free}}
   :=\{\rho_\omega:\omega\text{ is a state on }C_r^*(\F_K)\}.
\end{equation}
The state space of $C_r^*(\F_K)$ is weak-$*$ compact and convex, and the
map $\omega\mapsto\rho_\omega$ is weak-$*$ continuous into the
finite-dimensional space $M_K(\C)$.  Hence
$\mathcal K_K^{\mathrm{free}}$ is a compact convex subset of
$\Dens(\C^K)$.

For a self-adjoint element $b$ of a unital $C^*$-algebra, write
$\lambda_{\max}(b):=\max\spec(b)$.  Since the maximum of $\omega(b)$ over
all states $\omega$ equals $\lambda_{\max}(b)$, the support function of
$\mathcal K_K^{\mathrm{free}}$ is
\begin{equation}\label{eq:free-support-function}
  h_\infty(A)
  :=\max_{\rho\in\mathcal K_K^{\mathrm{free}}}\Tr(A\rho)
  =\lambda_{\max}(\mathcal T_\infty(A))
  =\frac1K\lambda_{\max}\!\left(
       \sum_{i,j=1}^K a_{ij}u_i^*u_j
     \right).
\end{equation}
Equivalently, $\mathcal K_K^{\mathrm{free}}$ is the intersection of the
half-spaces $\Tr(A\rho)\le h_\infty(A)$ over Hermitian $A$. The support function determines the minimum entropy through an exact Gibbs
variational formula.

\begin{proposition}[Free output-body entropy program]
\label{prop:gibbs-output-body}
With $\mathcal K_K^{\mathrm{free}}$ and $h_\infty$ as above,
\begin{equation}\label{eq:gibbs-output-body}
  \min_{\rho\in\mathcal K_K^{\mathrm{free}}}H(\rho)
  =\inf_{A=A^*\in M_K(\C)}
    \left\{\log\Tr(e^A)-h_\infty(A)\right\}.
\end{equation}
\end{proposition}

\begin{proof}
For every density matrix $\rho$, the Gibbs variational principle gives
\begin{equation}\label{eq:gibbs-dual-state}
   -H(\rho)
   =\sup_{A=A^*}
      \left\{\Tr(A\rho)-\log\Tr(e^A)\right\}.
\end{equation}
For full-rank $\rho$, equality is attained at $A=\log\rho$; the singular
case follows by sending the eigenvalues of $A$ on $\ker\rho$ to
$-\infty$.  Therefore
\begin{align*}
  -\min_{\rho\in\mathcal K_K^{\mathrm{free}}}H(\rho)
  &=\sup_{\rho\in\mathcal K_K^{\mathrm{free}}}
    \sup_{A=A^*}
    \left\{\Tr(A\rho)-\log\Tr(e^A)\right\}\\
  &=\sup_{A=A^*}
    \left\{h_\infty(A)-\log\Tr(e^A)\right\}.
\end{align*}
Taking negatives proves \eqref{eq:gibbs-output-body}.
\end{proof}

Writing
$A_0:=A-(\Tr A/K)I_K$, Haagerup's inequality gives the isotropic estimate
\begin{equation}\label{eq:isotropic-support-bound}
  h_\infty(A)
  \le\frac{\Tr A}{K}
    +\frac3K\|A_0\|_{\HS}.
\end{equation}
Indeed,
$\mathcal T_\infty(A)=(\Tr A/K)I+\mathcal T_\infty(A_0)$ and
$\|\mathcal T_\infty(A_0)\|\le3\|A_0\|_{\HS}/K$.
The present radius argument effectively replaces the anisotropic support
function $h_\infty$ by \eqref{eq:isotropic-support-bound}.  Determining or
sharply bounding the exact program \eqref{eq:gibbs-output-body} is therefore
the natural route to a substantially smaller output dimension without
contradicting Proposition~\ref{prop:reachable-no-go}.

\begin{remark}[Purity alone cannot improve the universal conversion]
\label{rem:purity-tight}
Fix $q\ge1$ and let $K>q$.  Consider
\[
   \rho_{K,q}
   :=\frac{I_{K-q}}{K-q}\oplus0_q.
\]
Then
\begin{align*}
  K\left\|\rho_{K,q}-\frac{I_K}{K}\right\|_{\HS}^2
  &=\frac{q}{K-q},\\
  \log K-H(\rho_{K,q})
  &=\log\frac{K}{K-q},\\
  H(\rho_{K,q})
  &=-\log\Tr(\rho_{K,q}^2)=\log(K-q).
\end{align*}
For fixed $q$,
\[
   \frac{\log K-H(\rho_{K,q})}
        {K\|\rho_{K,q}-I_K/K\|_{\HS}^2}
   \longrightarrow1
   \qquad(K\to\infty).
\]
Thus both the elementary quadratic entropy-deficit inequality and the
R\'enyi-$2$ lower bound used in Proposition~\ref{prop:one-copy} are sharp
on the full state space.  A substantial improvement must exploit
restrictions imposed by the orbit-Gram output body, rather than only a
stronger scalar inequality at fixed purity or Hilbert--Schmidt radius.
\end{remark}

\subsection{Outlook}

The present construction is asymptotic and algorithmic: it does not
provide a closed-form permutation tuple or an explicit finite threshold.
Moreover, Corollary~\ref{cor:finite-saturation} and
Proposition~\ref{prop:reachable-no-go} show that a substantial improvement
cannot come simply from obtaining a much smaller quadratic radius for the
finite tuple, or from restricting the same radius optimization to
directions realized by channel outputs.  Remark~\ref{rem:purity-tight}
also shows that the universal conversion from Hilbert--Schmidt distance,
or equivalently purity, to von Neumann entropy is already asymptotically
sharp on the full state space.  A natural next step is therefore to retain
more of the geometry of the output body.  In particular,
Proposition~\ref{prop:gibbs-output-body} reduces the corresponding free
problem to the variational formula \eqref{eq:gibbs-output-body}.
Determining this quantity, or obtaining useful bounds beyond the isotropic
estimate \eqref{eq:isotropic-support-bound}, may lead to substantially
better quantitative parameters.

A second question concerns finite-size explicitness.  The
O'Donnell--Wu spectral approximation in Lemma~\ref{lem:odw} is used here
only through its asymptotic fixed-parameter consequence.  Quantitative
bounds specialized to the scalar degree-two polynomials required in the
proof would yield explicit finite thresholds and independently verifiable
permutation tuples.  It would also be interesting to know whether the
algorithmic construction can be replaced by a closed-form algebraic
family satisfying the same simultaneous quadratic spectral estimates.

\section{Statements and Declarations}\label{sec:declarations}

\subsection{Independent concurrent work}

During the preparation of this manuscript, we became aware of the
independent work of Lovitz and Wu \cite{LovitzWu2026}, which also uses
deterministic O'Donnell--Wu permutation lifts to realize a
free-probabilistic mechanism for minimum-output-entropy nonadditivity at
$p=1$.  Accordingly, we do not claim priority for the general conclusion
that deterministic permutation constructions can yield
finite-dimensional $p=1$ counterexamples.  The two works, however,
implement the spectral derandomization at different interfaces.
Lovitz and Wu work with the random-subspace model and
retain finer information about the one-copy output geometry through
matrix-valued support polynomials, leading to substantially sharper
quantitative parameters.

The present work instead follows Collins's mixed-unitary route directly.
Its one-copy analysis reduces to a uniform family of scalar self-adjoint
degree-two polynomials controlled by the length-two Haagerup inequality,
while the exact Gram and purity identities identify the finite-dimensional
geometry encoded by this scalar certificate.  Restricting the deterministic
permutation tuple to the nontrivial standard representation produces real
orthogonal Stinespring blocks, so the usual product--conjugate Bell witness
becomes a genuine tensor-square violation with a real Stinespring
realization.  The subsequent sharpness, finite near-saturation, and
reachable-direction results delineate what this scalar-radius method can
achieve and isolate the full output-body support function and its Gibbs
entropy program as the natural finer level of analysis.  Thus the two
approaches are complementary: the random-subspace route retains finer
output geometry and gives stronger quantitative control, whereas the
present route emphasizes a direct scalar quadratic interface, an exact
finite-dimensional certificate, and a transparent real self-tensor
realization.

\subsection{Use of Artificial Intelligence}

The initial mathematical motivation for the present approach arose from the authors' observation that the free-probabilistic proof of minimum-output-entropy nonadditivity developed by Collins and related works, in particular its use of Haagerup's inequality for random mixed-unitary channels, appeared especially amenable to derandomization.  In exploring this possibility, the authors used a large language model (LLM) to assist with literature search; this process brought the deterministic spectral-approximation results of O'Donnell and Wu \cite{ODonnellWu2020} to our attention.

The LLM was subsequently used to help the authors analyze the technical interface between the random free-probabilistic argument and the deterministic O'Donnell--Wu construction, contributing to the development and organization of the approach presented in this paper.  It was also used to assist with the interactive refinement of the manuscript's mathematical exposition, readability, and presentation.  All mathematical arguments, interpretations, references, and final editorial decisions were independently checked and approved by the authors, who take full responsibility for the scientific content of the paper.

\bibliographystyle{alpha}
\bibliography{ref}

\end{document}